\documentclass[%
 reprint,
superscriptaddress,
 amsmath,amssymb,
 aps,
]{revtex4-2}

\usepackage{amsmath}
\usepackage {xcolor}
\usepackage{comment}
\usepackage{bbm}
\usepackage{mathrsfs}
\usepackage{comment}
\usepackage{amssymb}
\usepackage{graphicx}
\usepackage{subfigure}
\usepackage[colorlinks,
            linkcolor=blue,
            anchorcolor=blue,
            citecolor=blue]{hyperref}
\newtheorem{theorem}{Theorem}

\allowdisplaybreaks[4]

\def\ra{\rangle}
\def\la{\langle}

\usepackage{graphicx}
\usepackage{dcolumn}
\usepackage{bm}

\begin{document}

\preprint{APS/123-QED}

\title{Tighter entropic uncertainty relations in the presence of quantum memories for complete sets of mutually unbiased bases}

\author{Qing-Hua Zhang}
\email[]{qhzhang@csust.edu.cn}
\affiliation{School of Mathematics and Statistics, Changsha University of Science and Technology, Changsha 410114, China}
\affiliation{Hunan Provincial Key Laboratory of Mathematical Modeling and Analysis in Engineering, Changsha University of Sicence and Technology, Changsha 410114, China}

\author{Cong Xu}
\affiliation{School of Mathematical Sciences, Capital Normal University,
Beijing 100048, China}

\author{Jing-Feng Wu}
\affiliation{School of Mathematical and Big Data, Jining  University, Qufu 273155, China}

\author{Shao-Ming Fei}
\affiliation{School of Mathematical Sciences, Capital Normal University,
Beijing 100048, China}

\thanks{}

\begin{abstract}
Entropic uncertainty relations provide an information-theoretic framework for quantifying the fundamental indeterminacy inherent in quantum mechanics. We propose more stringent quantum-memory-assisted entropic uncertainty relations for complete sets of mutually unbiased bases in multipartite scenarios. We present lower and upper bounds of the quantum uncertainties based on the complementarity of the observables, the purity of the measured state, the (conditional) von-Neumann entropies, the Holevo quantities and mutual information. The results are illustrated by several representative cases, showing that our bounds are tighter than and outperform previously existing bounds.
\end{abstract}

\maketitle

\section{Introduction}

An essential aspect of quantum mechanics lies in the unavoidable indeterminacy that arises when dealing with complementary physical quantities. The original variance-based uncertainty relation related to the position and the momentum of a particle was proposed by Heisenberg~\cite{heisenberg1927uber}. Deutsch~\cite{deutsch1983uncertainty} established an entropic uncertainty relation for finite-spectrum observables within the Shannon entropy framework, which was subsequently tightened in later works by Kraus~\cite{kraus1987complementary} and by Maassen and Uffink~\cite{maassen1988generalized}:
\begin{equation}\label{mu}
H(M_1)+H(M_2)\geqslant -\log_2 c=: q_{MU},
\end{equation}
where $c=\max_{jk}|\langle \psi_j|\phi_k\rangle|^2$, with $|\psi_j\ra$ and $|\phi_k\ra$ being the eigenvectors of observables $M_1$ and $M_2$,  respectively. $H(M_1)=-\sum_i p_i\log_2 p_i$ ($H(M_2)=-\sum_i q_i\log_2 q_i$) is the Shannon entropy with $p_i=\langle \psi_i|\rho|\psi_i\rangle$ ($q_i=\langle \phi_i|\rho|\phi_i\rangle$). EURs for multiple measurements 
have also been extensively investigated ~\cite{sanchezruiz1995improved,sanchez1993entropic,coles2017entropic,PhysRevA.83.062338,chen2018improved,liu2015entropic,xiao2016strong,WANG2024138876,WANG2024129364}.

Recent advancement in the theory of entropic uncertainty relations is the extension to scenarios where the measured system is entangled with a quantum memory. In such settings, quantum correlations between the system and an environmental memory can lower the conditional entropy of measurement outcomes, which in turn relaxes the usual uncertainty bounds for an observer who has access to the memory. An entropic uncertainty relation in the presence of quantum memory was formulated by Renes $et\ al.$~\cite{renes2009conjectured} and by Berta $et\ al.$~\cite{berta2010uncertainty}, and later tested experimentally in Refs.~\cite{li2011experimental,prevedel2011experimental}. In the presence of quantum memory, the entropic uncertainty relation for two incompatible observables is given by
\begin{equation}\label{renesberta}
S(M_1|B)+S(M_2|B)\geqslant -\log_2 c+S(A|B),
\end{equation}
where $ S(M|B)=S(\rho^{MB})-S(\rho^B)$ denotes the conditional von Neumann entropy of the postmeasurement state $\rho^{MB}=\sum_i(|\psi_i\ra\la\psi_i|\otimes \mathbb{I}) \rho^{AB} (|\psi_i\ra\la\psi_i|\otimes \mathbb{I})$, $S(A|B)=S(\rho^{AB})-S(\rho^B)$, $\rho^B$ is the reduced state of particle $B$, and $S(\rho)=-\mathrm{tr}\rho\log\rho$ is von Neumann entropy. The lower bound on measurement uncertainty is influenced by the amount of entanglement shared between the measured system \(A\) and the quantum memory \(B\). When the memory \(B\) is absent, inequality~(\ref{renesberta}) reduces to
$H(M_1)+H(M_2)\geqslant -\log_2 c + S(\rho^A)$,
which yields a tighter lower bound than inequality~(\ref{mu}) for states satisfying \(S(\rho^A)>0\). Consequently, QMA-EURs play an important role in a variety of quantum information processing tasks, such as quantum key distribution~\cite{koashi2009simple},
quantum cryptography~\cite{dupuis2014entanglement,konig2012unconditional},
quantum randomness~\cite{vallone2014quantum}, entanglement witness~\cite{berta2014entanglement,huang2010entanglement,hu2012quantum}, EPR steering~\cite{sun2018demonstration,schneeloch2013einstein} and quantum metrology~\cite{giovannetti2011advances}, tighter lower bounds of QMA-EURs have attracted much attention ~\cite{pati2012quantum,coles2014improved,adabi2016tightening,zhang2015entropic,hu2013competition,ming2020improved,xie2021optimized,dolatkhah2020tightening,dolatkhah2022tripartite}. A tripartite entropic uncertainty relation involving two quantum memories $B$ and $C$ was developed by Renes \textit{et al.}~\cite{renes2009conjectured} and by Berta \textit{et al.}~\cite{berta2010uncertainty}.

\begin{equation}\label{reneslb}
S(M_1 | B)+S(M_2 | C) \geqslant q_{MU}.
\end{equation}

Later, Ming $et\ al.$~\cite{ming2020improved} improved the tripartite QMA-EUR in terms of the Holevo quantity and mutual information,
\begin{equation}
S(M_1 | B)+S(M_2 | C) \geqslant q_{M U}+\max \left\{0, \delta_1\right\},
\end{equation}
where $\delta_1=2 S(A)+q_{M U}-\mathcal{I}(A:B)-\mathcal{I}(A:C)+\mathcal{I}(M_2:B)+\mathcal{I}(M_1: C)-H(M_1)-H(M_2)$, $\mathcal{I}(A:B)=S\left(\rho^A\right)+S\left(\rho^B\right)-S\left(\rho^{A B}\right)$ and $\mathcal{I}(A:C)=S\left(\rho^A\right)+S\left(\rho^C\right)-S\left(\rho^{A C}\right)$ stand for the mutual information, $\mathcal{I}(M_2: B)=S\left(\rho^{M_2}\right)+S\left(\rho^B\right)-S\left(\rho^{M_2 B}\right)$ and $\mathcal{I}(M_1: C)=S\left(\rho^{M_1}\right)+S\left(\rho^C\right)-S\left(\rho^{M_1 C}\right)$ are the Holevo quantities. Wu  $et\ al.$~\cite{wu2022tighter} improved the above bound further,
\begin{equation}
S(M_1 | B)+S(M_2 | C) \geqslant q_{M U}+\max \left\{0, \delta_2\right\},
\end{equation}
where $\delta_2=2 S(A)+q_{M U}-\mathcal{I}(M_1: B)- \mathcal{I}(M_2: C)-H(M_1)-H(M_2)$.

In practical quantum information processing, it is often necessary to characterize measurement uncertainty beyond pairs of observables in bipartite or tripartite scenarios, extending instead to multiple measurements performed on correlated multipartite systems~\cite{wu2022tighter,zhang2023entropic,XU2025130570,PhysRevA.98.032329}. By incorporating conditional von Neumann entropies together with mutual information and Holevo quantities, Zhang and Fei derived an entropic uncertainty relation for $m$-tuple of measurements $\mathbf{M}={\{M_i\},\ i=1,\dots,m}$, in the context of $n$ memories $(n\leqslant m)$ ~\cite{zhang2023entropic}:
\begin{equation}\label{zhang}
\begin{aligned}
\sum_{t=1}^n\sum_{M_i\in \mathbf{S}_t} S(M_i|B_t)\geqslant &-\frac{1}{m-1} \log _2\left(\prod_{i<j}^{m} c_{ij}\right)\\
&+\frac{1}{m-1}\sum_t \frac{m_t(m_t-1)}{2}S(A|B_t)\\
&+\max\left\{0,\delta_{mn}\right\},
\end{aligned}
\end{equation}
where
\begin{align*}
\delta_{mn}=&\frac{m(m-1)-\sum_{t=1}^n m_t(m_t-1)}{2(m-1)}S(A)\\
&+\sum_{t=1}^n\frac{m_t(m_t-1)}{2(m-1)}\mathcal{I}(A:B_t)\\
&-\sum_{t=1}^n\sum_{M_i\in \mathbf{S}_t} \mathcal{I}(M_i:B_t),
\end{align*}
 $c_{ij}=\max _{k,l}|\langle \psi_k^i |\psi_l^j\rangle|^2$ with $|\psi_k^i\rangle$ and $|\psi_l^j\rangle$ the eigenvectors of ${M}_i$ and $M_j$, respectively, the $n$ non-empty subsets $\mathbf{S}_t$ of $\mathbf{M}$ satisfy $\bigcup_{t=1}^n \mathbf{S}_t=\mathbf{M}$ and $\mathbf{S}_s \bigcap \mathbf{S}_t=\emptyset$ for $s\neq t$, and $m_t$ is the cardinality of $\mathbf{S}_t$.

Mutually unbiased bases (MUBs) constitute an essential issue in quantum information theory. As shown in Refs.~\cite{xie2021optimized,zhang2023entropic}, the corresponding uncertainty bounds are governed by the degree of complementarity between the measured observables, as well as by information-theoretic quantities such as the (conditional) von Neumann entropy, the Holevo quantity, and the mutual information. In this work, we focus on proposing more stringent QMA-EURs for mutually unbiased bases (MUBs).

\section{QMA-EURs for complete sets of MUBs}
Let $|\psi_j\ra$ and $|\phi_k\ra$ be the eigenvectors of observables $M_1$ and $M_2$,  respectively. 
The $d$-dimensional  orthonormal bases $\{|\phi_k\rangle\}$ and $\{|\psi_j\rangle\}$ are mutually unbiased bases if
$|\langle \psi_j|\phi_k\rangle|^2=\frac{1}{d}$ for all $j$ and $k$. For simplicity, we denote $M$ also the corresponding
basis given by the  eigenvectors of $M$ without confusion. A collection of \(n\) orthonormal bases \(\{M_i\}\) is referred to as \(n\) MUBs if any two distinct bases \(M_j\) and \(M_k\) satisfy the mutual unbiasedness condition for all \(j \neq k\). 
For instance, the eigenvectors of the three standard Pauli operators,
$\sigma_{{x}}  =|0\rangle\langle 1|+|1\rangle\langle 0|$, $\sigma_{{y}}  =-\mathbf{i}|0\rangle\langle 1|+\mathbf{i}|1\rangle\langle 0|$ and $\sigma_{{z}}  =|0\rangle\langle 0|-|1\rangle\langle 1|$, form a set of three MUBs, where $\mathbf{i}=\sqrt{-1}$ is the imaginary unit. It has been established that when the Hilbert space dimension \(d\) is a power of a prime, the largest possible number of MUBs equals \(d+1\) \cite{10.1063/1.2713445}. Such a maximal collection is commonly referred to as a complete set of mutually unbiased bases (CMUBs).

Consider an uncertainty game involving $n+1$ players: Alice, Bob$_1$, Bob$_2$, $\dots$, Bob$_n$. The players share an $(n+1)$-partite quantum state. Alice randomly selects one measurement from CMUBs $\mathbf{M} = \{M_i \mid i = 1, \dots, d+1\}$. Partition $\mathbf{M}$ into $n$ non-empty, mutually disjoint subsets $\mathbf{S}_t \subset \mathbf{M}$ ($t = 1, \dots, n$) such that $\bigcup_{t=1}^n \mathbf{S}_t = \mathbf{M}$. If Alice chooses a measurement from the subset $\mathbf{S}_t$, she announces her choice to Bob$_t$. Each Bob$_t$ attempts to guess the outcome of Alice’s measurement, see Fig. \ref{fig_game}. Based on the conditional von-Neumann entropies, purity and Holevo quantities, we have the following theorem.
 \begin{figure}[t]
		\includegraphics[width=8cm]{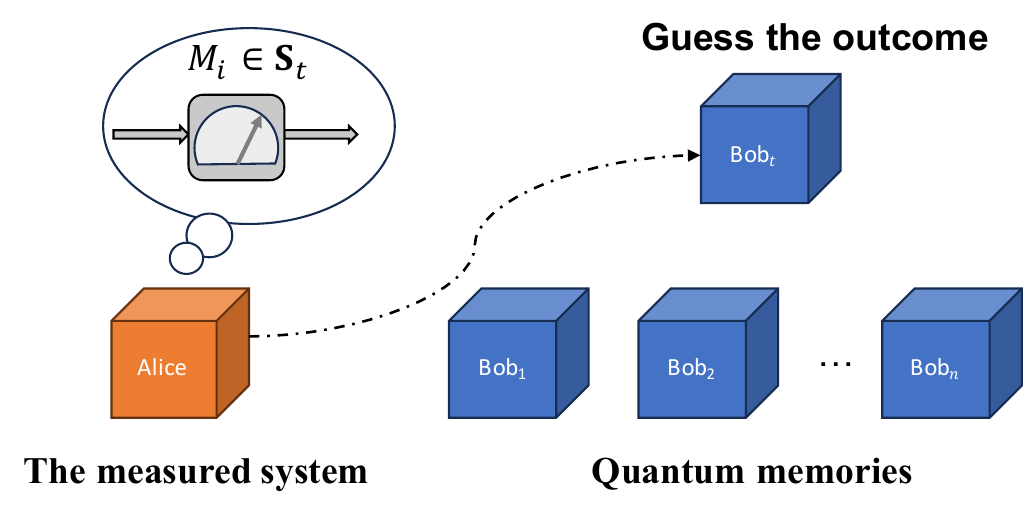}
 \caption{\label{fig_game}Illustration of the uncertainty game. Alice and all Bobs share a common state and agree with CMUBs $\mathbf{M} = \{M_i \mid i = 1, \dots, d+1\}$. Alice carries out one measurement $M_i$ of $\mathbf{M}$ and announces her choice to Bob$_t$ if  $M_i\in \mathbf{S}_t$. The task of each Bob is to guess the Alice's measurement outcome.}
 \end{figure}

\begin{theorem}\label{qmaeurthm1}
The following entropic uncertainty relation for CMUBs in the context of $n$ memories holds,
\begin{equation}
\begin{aligned}\label{Thm1}
&\sum_{t=1}^n\sum_{M_i\in \mathbf{S}_t} S(M_i|B_t) \\ 
&\geqslant \frac{d+1}{2} \log _2 d+\sum_t \frac{m_t(m_t-1)}{2d}S(A|B_t)\\
&+\max\left\{0,\delta_{n}^{CMUBs}\right\},
\end{aligned}
\end{equation}
where
\begin{align*}
\delta_{n}^{CMUBs}=&L_{CMUBs}-\frac{d+1}{2} \log _2 d-\sum_{t=1}^n\frac{m_t(m_t-1)}{2d}S(A)\\
&+\sum_{t=1}^n\frac{m_t(m_t-1)}{2d}\mathcal{I}(A:B_t)\\
&-\sum_{t=1}^n\sum_{M_i\in \mathbf{S}_t} \mathcal{I}(M_i:B_t),
\end{align*}
$m_t$ is the cardinality of $\mathbf{S}_t$, $L_{CMUBs}=(d+1)[\log_2(1+\lfloor{v}\rfloor)-\frac{\lfloor{v}\rfloor}{v}(1+\lfloor{v}\rfloor-v)\log_2(1+\frac{1}{\lfloor{v}\rfloor})]$, $\lfloor{\cdot}\rfloor$ denotes the integer part of a real number, $v=\frac{d+1}{\Pi(\rho^A)+1}$ with the purity $\Pi(\rho^A)={\mathrm{tr}}((\rho^A)^2)$, $\rho^A$ is the reduced state obtained by tracing over the Bobs' systems.
\end{theorem}

{\textit{Proof.}} Using the relation between the conditional von Neumann entropy and the mutual information,
\(S(M_i|B_t)=H(M_i)-\mathcal{I}(M_i:B_t)\),
one obtains

\begin{equation}\label{conditionentropy}
\begin{aligned}
\sum_{t=1}^n\sum_{M_i\in \mathbf{S}_t} S(M_i|B_t)=\sum_{i=1}^{d+1} H(M_i)-\sum_{t=1}^n\sum_{M_i\in \mathbf{S}_t} \mathcal{I}(M_i:B_t).
\end{aligned}
\end{equation}

For CMUBs it has been shown that \cite{sanchezruiz1995improved,sanchez1993entropic,coles2017entropic}, 
\begin{equation}\label{sanchezruiz}
\begin{aligned}
\sum_{i=1}^{d+1} H (M_i)\geqslant & (d+1)[\log_2(1+\lfloor{v}\rfloor)\\
&-\frac{\lfloor{v}\rfloor}{v}(1+\lfloor{v}\rfloor-v)\log_2(1+\frac{1}{\lfloor{v}\rfloor})],
\end{aligned}
\end{equation}
where $v=\frac{d+1}{{\mathrm{tr}}(\rho^2)+1}$ with $\rho$ being the state measured. 

By incorporating the generalized entropic uncertainty relation for \(d+1\) CMUBs in the presence of \(n\) quantum memories, as established in Ref.~\cite{zhang2023entropic},
\begin{equation}\label{zhangeur1}
\begin{aligned}
\sum_{t=1}^n\sum_{M_i\in \mathbf{S}_t} S(M_i|B_t)\geqslant &\frac{d+1}{2} \log _2 d\\ 
&+\sum_{t=1}^{n} \frac{m_t(m_t-1)}{2d}S(A|B_t),
\end{aligned}
\end{equation}
we complete the proof. $\Box$

In particular, if only two players ($n=1$), Alice and Bob, agree on CMUBs $\left\{M_i\right\}_{i=1}^{d+1}$. Alice performs one of the prescribed measurements and publicly communicates her choice to Bob, whose goal is to infer the corresponding measurement outcome. In the bipartite scenario, where \(d+1\) measurements are implemented on Alice’s subsystem \(A\), the uncertainty relation~(\ref{Thm1}) reduces to

\begin{equation}
\begin{aligned}
\sum_{i=1}^{d+1}S(M_i|B)\geqslant &\frac{d+1}{2} \log _2 d+\frac{d+1}{2}S(A|B)\\
&+\max\left\{0,\delta_{1}^{CMUBs}\right\},
\end{aligned}
\end{equation}
where 
\begin{equation*}
\begin{aligned}
\delta_{1}^{CMUBs}=&L_{CMUBs}-\frac{d+1}{2} \log _2 d-\frac{d+1}{2}S(A)\\
&+\frac{d+1}{2}\mathcal{I}(A:B)-\sum_{i=1}^{d+1} \mathcal{I}(M_i:B).
\end{aligned}
\end{equation*}

For the case of  $d+2$ players ($n=d+1$),  the cardinality $m_t=1$. Our QMA-EUR in Theorem \ref{qmaeurthm1} reduces to
\begin{equation}\label{mubsd2}
\begin{aligned}
&\sum_{i=1}^{d+1} S(M_i|B_i)\geqslant \frac{d+1}{2} \log _2 d+\max\left\{0,\delta_{d+1}^{CMUBs}\right\},
\end{aligned}
\end{equation}
where
\begin{equation*}
\begin{aligned}
\delta_{d+1}^{CMUBs}=L_{CMUBs}-\frac{d+1}{2} \log _2 d
-\sum_{i=1}^{d+1}\mathcal{I}(M_i:B_i).
\end{aligned}
\end{equation*}

In Refs.~\cite{sanchezruiz1995improved,sanchez1993entropic}, Sánchez-Ruiz proposed a tighter entropic uncertainty relation for CMUBs with respect to a measured state $\rho$,
\begin{equation}\label{sanchezruiz2}
\sum_{i=1}^{d+1} H\left(M_i\right)\leqslant {U}_{CMUBs},
\end{equation}
where ${U}_{CMUBs}=(d+1)\log_2 d-\frac{(d-1)(d\mathrm{tr}(\rho^2)-1)\log_2(d-1)}{d(d-2)}$ for $d>2$ and $U_{CMUBs}=(d+1)\log_2 d-\frac{(d-1)(d\mathrm{tr}(\rho^2)-1)}{d\ln 2}$ for $d=2$.

\begin{theorem}\label{qmaeurthm2}
The following entropic uncertainty relation in terms of the conditional von-Neumann entropies, purity and Holevo quantities holds for CMUBs in the context of $n$ memories,
\begin{equation}
\begin{aligned}
\sum_{t=1}^n\sum_{M_i\in \mathbf{S}_t} S(M_i|B_t) 
\leqslant U_{CMUBs}-\sum_{t=1}^n\sum_{M_i\in \mathbf{S}_t} \mathcal{I}(M_i:B_t),
\end{aligned}
\end{equation}
where $U_{CMUBs}=(d+1)\log_2 d-\frac{d-1}{d(d-2)}\log_2(d-1)(d\Pi(\rho^A)-1)$ for $d>2$ and $U_{CMUBs}=(d+1)\log_2 d-\frac{d-1}{d\ln 2}(d(\Pi(\rho^A)-1)$ for $d=2$.
\end{theorem}

{\textit{Proof.}} By combining two inequalities (\ref{conditionentropy}) and (\ref{sanchezruiz2}), we obtain
\begin{align*}
&\sum_{t=1}^n\sum_{M_i\in \mathbf{S}_t} S(M_i|B_t)\\
&=\sum_{i=1}^{d+1} H(M_i)-\sum_{t=1}^n\sum_{M_i\in \mathbf{S}_t} \mathcal{I}(M_i:B_t)\\
&\leqslant U_{CMUBs}-\sum_{t=1}^n\sum_{M_i\in \mathbf{S}_t} \mathcal{I}(M_i:B_t),
\end{align*}
which completes the proof. $\Box$

To illustrate our results, we consider below particular cases with detailed examples.

\subsection{One quantum memory}

{\textit{Example 1.}}  Alice and Bob share a class of two-qubit pure states, $\rho=\cos{\theta}|00\ra+\sin{\theta}|11\ra$,
where $0\leqslant \theta\leqslant 2\pi$. We consider eigenvectors of three Pauli matrices $\sigma_x,\ \sigma_y$ and $\sigma_z$ as CMUBs,
\begin{equation}\label{cmubs2}
\begin{aligned}
&M_{{1}}  =\left\{|0\rangle,|1\rangle\right\}, \\
&M_{{2}} =\{\frac{|0\rangle+|1\rangle}{\sqrt{2}},\frac{|0\rangle-|1\rangle}{\sqrt{2}}\}, \\
&M_{{3}}  =\{\frac{|0\rangle+\mathbf{i}|1\rangle}{\sqrt{2}},\frac{|0\rangle-\mathbf{i}|1\rangle}{\sqrt{2}}\}.
\end{aligned}
\end{equation}
In this case, our (\ref{Thm1}) reduces to
\begin{equation}
\begin{aligned}
&\sum_{i=1}^3 S(M_i|B)\geqslant \frac{3}{2} +\frac{3}{2}S(A|B)+\max\left\{0,\delta_{1}^{CMUBs}\right\},
\end{aligned}
\end{equation}
where $\delta_{1}^{CMUBs}=L_{CMUBs}-\frac{3}{2} -\frac{3}{2}S(A)+\frac{3}{2}\mathcal{I}(A:B)-\sum_{i=1}^{3} \mathcal{I}(M_i:B)$.
Our theorem~\ref{qmaeurthm2} gives rise to
\begin{equation}
\begin{aligned}
&\sum_{i=1}^3 S(M_i|B)\leqslant 3-\frac{2\mathrm{tr}((\rho^A)^2)-1}{2\ln 2} -\sum_{i=1}^{3} \mathcal{I}(M_i:B).
\end{aligned}
\end{equation}

As illustrated in Fig.~\ref{figex1}, the bound provided by Theorem~\ref{qmaeurthm1} is strictly tighter than the corresponding bound in inequality~(\ref{zhang}) reported in Ref.~\cite{zhang2023entropic}.

 \begin{figure}[tbp]
\includegraphics[width=9cm]{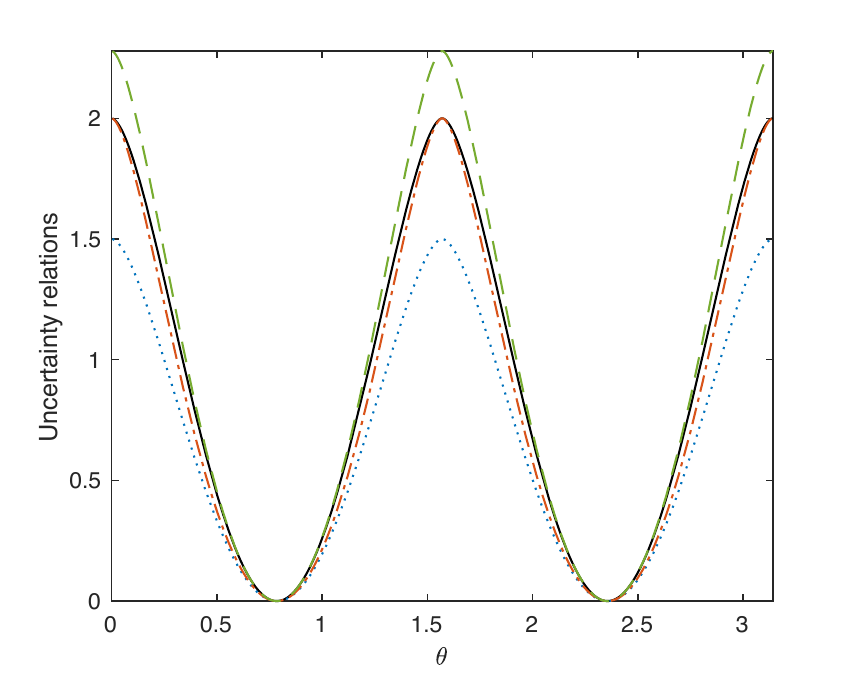}
 \caption{ \label{figex1}Comparisons between our uncertainty relations with (\ref{zhang}) in the scenario of one memory. The black (solid) curve represents the uncertainty. The blue (dotted) and red (dot-dashed) denote the lower bounds in (\ref{zhang}) and Theorem~\ref{qmaeurthm1}, respectively. The green (dashed) curve denotes the upper bound of Theorem~\ref{qmaeurthm2}.}
 \end{figure}

{\textit{Example 2.}}  Let us consider a class of two-qutrit pure states, $\rho=\sin{\phi}\cos{\theta}|00\ra+\sin{\phi}\sin{\theta}|11\ra+\cos{\phi}|22\ra$,
where $0\leqslant \phi\leqslant \pi$ and $0\leqslant \theta\leqslant  2\pi$. We take into account the following CMUBs in this case \cite{10.5555/2011464.2011470},
\begin{align*}
&M_{{1}}  \!=\!\{|0\rangle,|1\rangle,|2\ra\}, \\
&M_{{2}} \!=\!\{\frac{|0\rangle\!+\!|1\rangle\!+\!|2\rangle}{\sqrt{3}},\frac{|0\rangle\!+\!\omega|1\rangle\!+\!\omega^2|2\rangle}{\sqrt{3}},\frac{|0\rangle\!+\!\omega^2|1\rangle\!+\!\omega|2\rangle}{\sqrt{3}}\}, \\
&M_{{3}}  \!=\!\{\frac{|0\rangle\!+\!\omega|1\rangle\!+\!\omega|2\rangle}{\sqrt{3}},\frac{|0\rangle\!+\!\omega^2|1\rangle\!+\!|2\rangle}{\sqrt{3}},\frac{|0\rangle\!+\!|1\rangle\!+\!\omega^2|2\rangle}{\sqrt{3}}\}, \\
&M_{{4}} \! =\!\{\frac{|0\rangle\!+\!\omega^2|1\rangle\!+\!\omega^2|2\rangle}{\sqrt{3}},\frac{|0\rangle\!+\!\omega|1\rangle\!+\!|2\rangle}{\sqrt{3}},\frac{|0\rangle\!+\!|1\rangle\!+\!\omega|2\rangle}{\sqrt{3}}\},
\end{align*}
where $\omega=\frac{-1+\sqrt{3}\mathbf{i}}{2}$.

 In this scenario, our Theorem~\ref{qmaeurthm1} gives rise to
\begin{equation}
\begin{aligned}
&\sum_{i=1}^4 S(M_i|B)\geqslant 2\log_2 3 +2 S(A|B)+\max\left\{0,\delta_{1}^{CMUBs}\right\},
\end{aligned}
\end{equation}
where 
$\delta_{1}^{CMUBs}=L_{CMUBs}-2\log_2 3 -2 S(A)+2\mathcal{I}(A:B)-\sum_{i=1}^{4} \mathcal{I}(M_i:B)$. Our Theorem~\ref{qmaeurthm2} reduces to
\begin{equation}
\begin{aligned}
&\sum_{i=1}^4 S(M_i|B)\leqslant 4\log_2 3+\frac{2}{3}-2\mathrm{tr}((\rho^A)^2)-\sum_{i=1}^{4} \mathcal{I}(M_i:B).
\end{aligned}
\end{equation}

We compare the lower bounds of our theorems with that of ~(\ref{zhang}) by setting $\phi=\frac{\pi}{4}$ and $\theta=\frac{\pi}{4}$, respectively. Fig.~\ref{figex2} demonstrates that the lower bound established in Theorem~\ref{qmaeurthm1} surpasses the bound given in inequality~(\ref{zhang}).

 \begin{figure}[tbp]
	\begin{subfigure}
	 \centering
	\label{figex2a}
		\includegraphics[width=4.1cm]{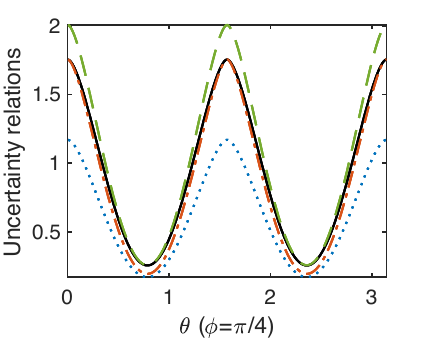}
	\end{subfigure}
	\begin{subfigure}
		\centering
		\label{figex2b}
		\includegraphics[width=4.1cm]{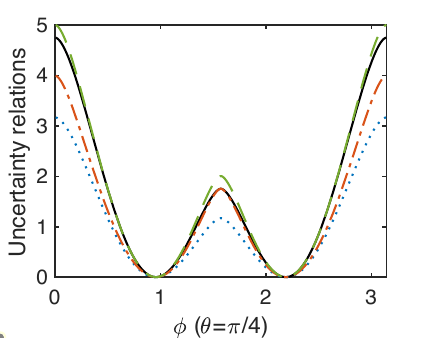}
	\end{subfigure}
 \caption{Comparisons between our uncertainty relations with (\ref{zhang}) in the scenario of one memory. The black (solid) curve represents the uncertainty. The blue (dotted) and red (dot-dashed) denote the lower bounds in (\ref{zhang}) and Theorem~\ref{qmaeurthm1}, respectively. The green (dashed) curves denotes the upper bound of Theorem~\ref{qmaeurthm2}.}
 \label{figex2}
 \end{figure}

 {\textit{Example 3.}} To address a more general situation, we consider arbitrary random states acting on a \(4\otimes 4\) Hilbert space, expressed as
\[
\rho^{AB}=\sum_{k=1}^{16} p_k |\psi_k\ra\la\psi_k|,
\]
where \(p_k\) and \({\psi_k}\) represent the eigenvalues and eigenvectors of \(\rho^{AB}\), respectively. The probability distribution \(\{p_k\}\) is generated by first producing a sequence of random numbers \(q_k\) using a uniform random function \(f(0,1)\) on the interval \([0,1]\), with \(q_1=f(0,1)\) and \(q_{k+1}=f(0,1)\,q_k\) for \(k=1,\ldots,15\). The normalized probabilities are then defined as \(p_k=q_k/\sum_{j=1}^{16} q_j\).

To construct the corresponding eigenvectors, we generate a real matrix \(R\) of order sixteen whose entries are independently drawn from the uniform distribution \(f(-1,1)\) on \([-1,1]\). From this matrix, a random Hermitian matrix is formed as
\[
\tilde{R}=D+(U^{\mathsf{T}}+U)+\mathbf{i}(L^{\mathsf{T}}+L),
\]
where \(D\), \(U\), and \(L\) denote the diagonal, strictly upper triangular, and strictly lower triangular parts of \(R\), respectively, and \(U^{\mathsf{T}}\) and \(L^{\mathsf{T}}\) are their transposes. Diagonalizing \(\tilde{R}\) yields sixteen normalized eigenvectors, which are used to define the random state \(\rho^{AB}\).

For $d=4$, we consider the following complete MUBs \cite{10.5555/2011464.2011470}:
\begin{align*} 
& M_1=\left\{(1,0,0,0),(0,1,0,0),(0,0,1,0),(0,0,0,1)\right\}, \\ 
& M_2=\{\frac{1}{2}(1,1,1,1), \frac{1}{2}(1,1,-1,-1),\\
&\quad\quad\quad\ \frac{1}{2}(1,-1,-1,1), \frac{1}{2}(1,-1,1,-1)\}, \\ 
& M_3=\{\frac{1}{2}(1,-1,-\mathbf{i},-\mathbf{i}), \frac{1}{2}(1,-1, \mathbf{i}, \mathbf{i}),\\
&\quad\quad\quad\ \frac{1}{2}(1,1, \mathbf{i},-\mathbf{i}), \frac{1}{2}(1,1,-\mathbf{i}, \mathbf{i})\}, \\ & 
M_4=\{\frac{1}{2}(1,-\mathbf{i},-\mathbf{i},-1), \frac{1}{2}(1,-\mathbf{i}, \mathbf{i}, 1),\\
&\quad\quad\quad\ \frac{1}{2}(1, \mathbf{i}, \mathbf{i},-1), \frac{1}{2}(1, \mathbf{i},-\mathbf{i}, 1)\}, \\ 
& M_5=\{\frac{1}{2}(1,-\mathbf{i},-1,-\mathbf{i}), \frac{1}{2}(1,-\mathbf{i}, 1, \mathbf{i}), \\
&\quad\quad\quad\ \frac{1}{2}(1, \mathbf{i},-1, \mathbf{i}), \frac{1}{2}(1, \mathbf{i}, 1,-\mathbf{i})\}.
\end{align*}

 In this scenario, our Theorem~\ref{qmaeurthm1} becomes
\begin{equation}
\begin{aligned}
&\sum_{i=1}^5 S(M_i|B)\geqslant \frac{5}{2}\log_2 5 +\frac{5}{2}S(A|B)+\max\left\{0,\delta_{1}^{CMUBs}\right\},
\end{aligned}
\end{equation}
where 
$\delta_{1}^{CMUBs}=L_{CMUBs}-\frac{5}{2}\log_2 5 -\frac{5}{2} S(A)+\frac{5}{2}\mathcal{I}(A:B)-\sum_{i=1}^{5} \mathcal{I}(M_i:B)$. Our Theorem~\ref{qmaeurthm2} reduces to
\begin{equation}
\begin{aligned}
\sum_{i=1}^5 S(M_i|B)&\leqslant 10-\frac{3\log_2 3}{8}(4\mathrm{tr}((\rho^A)^2)-1)\\
&\quad-\sum_{i=1}^{5} \mathcal{I}(M_i:B).
\end{aligned}
\end{equation}
The Fig.~\ref{figex53} indicates that the uncertainty bound derived in Theorem~\ref{qmaeurthm1} improves upon the bound presented in inequality~(\ref{zhang}).

 \begin{figure}[tbp]
		 \centering
 \subfigure[]
 {
\label{figex5a3}
 \includegraphics[width=4.1cm]{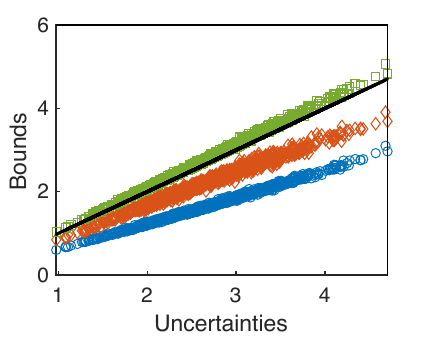}}
 \subfigure[]
 {
  \label{figex5b3}
 \includegraphics[width=4.1cm]{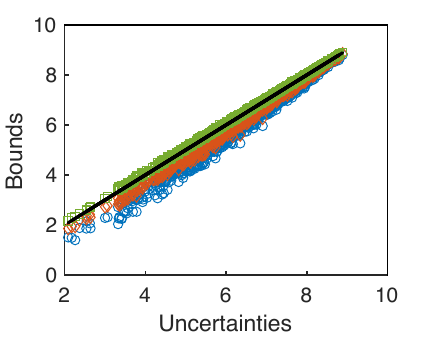}}
 \caption{Comparisons between our uncertainty relations and (\ref{zhang}) in the scenario of three memories with respect to $10^3$ randomly pure states (a) and  $10^3$ randomly mixed states (b). The $y$ axis denotes bounds of the uncertainty. The $x$ axis denotes uncertainty. The blue circles and red diamonds denote the lower bounds of (\ref{zhang}) and Theorem~\ref{qmaeurthm1}, respectively. The green squares denote the upper bound of Theorem~\ref{qmaeurthm2}. }
  \label{figex53}
 \end{figure}

 \subsection{Two quantum memories}
Consider that Alice, Bob$_1$ and Bob$_2$ agree on CMUBs $\left\{M_i\right\}_{i=1}^{d+1}\equiv \mathbf{S}_1 \bigcup\mathbf{S}_2$ with $\mathbf{S}_1 \bigcap \mathbf{S}_2=\emptyset$. If Alice carries out the MUB measurement in $\mathbf{S}_1$, she announces her choice to Bob$_1$, else announces to Bob$_2$. Bob$_1$ and Bob$_2$'s task is to guess the outcome of Alice's measurement.

{\textit{Example 4.}}  Alice, Bob$_1$ and Bob$_2$ share the generalized W states,
$|W\ra=\sin{\phi}\cos{\theta}|001\ra+\sin{\phi}\sin{\theta}|010\ra+\cos{\phi}|100\ra$
where $0\leqslant\phi\leqslant\pi$ and $0\leqslant\theta\leqslant 2\pi$. Alice carries out the measurement given by the CMUBs in (\ref{cmubs2}). Bob$_1$ guesses the outcome when Alice carries out $M_1$ and Bob$_2$ guesses the outcome when Alice carries out $M_2$ and $M_3$. In this scenario, our Theorem~\ref{qmaeurthm1} gives rise to
\begin{equation}
\begin{aligned}
& S(M_1|B_1)+S(M_2|B_2)+S(M_3|B_2) \\ 
&\geqslant \frac{3}{2}+ \frac{1}{2}S(A|B_2)+\max\left\{0,\delta_{2}^{CMUBs}\right\},
\end{aligned}
\end{equation}
where $\delta_{2}^{CMUBs}=L_{CMUBs}-\frac{3}{2} - \frac{1}{2}S(A)+\frac{1}{2}\mathcal{I}(A:B_2)-\mathcal{I}(M_1:B_1)-\mathcal{I}(M_2:B_2)-\mathcal{I}(M_3:B_2)$.
Our theorem~\ref{qmaeurthm2} reduces to
\begin{equation}
\begin{aligned}
& S(M_1|B_1)+S(M_2|B_2)+S(M_3|B_2) \\ 
&\leqslant 3\!-\!\frac{2\mathrm{tr}((\rho^A)^2)\!-\!1}{2\ln 2}\!-\!\mathcal{I}(M_1\!:\!B_1)\!-\!\mathcal{I}(M_2\!:\!B_2)\!-\!\mathcal{I}(M_3\!:\!B_2).
\end{aligned}
\end{equation}

We compare the lower bounds of our theorem with that of ~(\ref{zhang}) by taking $\phi=\frac{2\pi}{3}$ and $\theta=\frac{2\pi}{3}$, respectively. As shown in Fig.~\ref{figex3}, the lower bound obtained from Theorem~\ref{qmaeurthm1} is stronger than the bound given by inequality~(\ref{zhang}).

\begin{figure}[tbp]
	\begin{subfigure}
	 \centering
	\label{figex3a}
		\includegraphics[width=4.1cm]{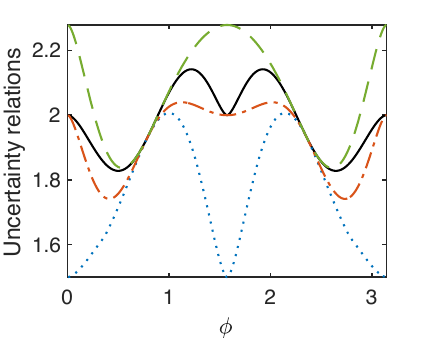}
	\end{subfigure}
	\begin{subfigure}
		\centering
		\label{figex3b}
		\includegraphics[width=4.1cm]{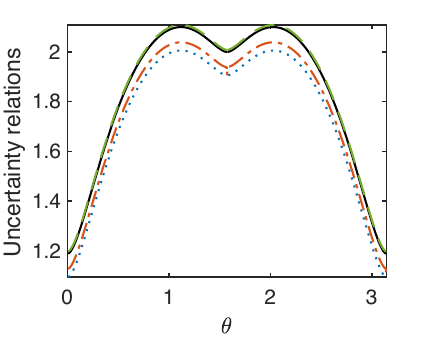}
	\end{subfigure}
 \caption{Comparisons between our uncertainty relations and (\ref{zhang}) in the scenario of two memories. The black (solid) curve represents the uncertainty. The blue (dotted) and red (dot-dashed) denote the lower bounds in (\ref{zhang}) and Theorem~\ref{qmaeurthm1}, respectively. The green (dashed) curve denotes the upper bound of Theorem~\ref{qmaeurthm2}.}
 \label{figex3}
 \end{figure}
 
\subsection{Three quantum memories}

Consider that four players Alice, Bob$_1$, Bob$_2$ and Bob$_3$ agree on the CMUBs given by the three Pauli matrices $\sigma_x,\ \sigma_y$ and $\sigma_z$. Alice carries out the $i$-th MUB and announces her choice to Bob$_i$. Bob$_i$'s task is to minimize the uncertainty about guessing the outcome of Alice's measurement.

In this case, our Theorem~\ref{qmaeurthm1} reduces to
\begin{equation}
\begin{aligned}
&\sum_{i=1}^3 S(M_i|B_i)\geqslant \frac{3}{2} +\max\left\{0,\delta_{3}^{CMUBs}\right\},
\end{aligned}
\end{equation}
where $\delta_{3}^{CMUBs}=L_{CMUBs}-\frac{3}{2}  -\sum_{i=1}^{3} \mathcal{I}(M_i:B_i)$.
Our Theorem~\ref{qmaeurthm2} reduces to
\begin{equation}
\begin{aligned}
&\sum_{i=1}^3 S(M_i|B_i)\leqslant 3-\frac{2\mathrm{tr}((\rho^A)^2)-1}{2\ln 2} -\sum_{i=1}^{3} \mathcal{I}(M_i:B_i).
\end{aligned}
\end{equation}

{\textit{Example 5.}}   Alice, Bob$_1$, Bob$_2$ and Bob$_3$ share four-qubit pure states, $\rho=\cos{\theta}|0000\ra+\sin{\theta}|1111\ra$, where $0\leqslant \theta\leqslant 2\pi$. As shown in Fig.~\ref{figex4} the lower bound of our Theorem~\ref{qmaeurthm1} is stronger than the bound given by inequality~(\ref{zhang}).

{\textit{Example 6.}} The set of $4\otimes 4$ random states generated in Example 3 can be regarded as a set of four-qubit states. When consider the CMUBs given by the three Pauli matrices, the Fig.~\ref{figex5} shows that the lower bound of our Theorem~\ref{qmaeurthm1} is also tighter than that of ~(\ref{zhang}) in the scenario of $10^3$ four-qubit random states.

\begin{figure}[t]
	 \centering
		\includegraphics[width=9cm]{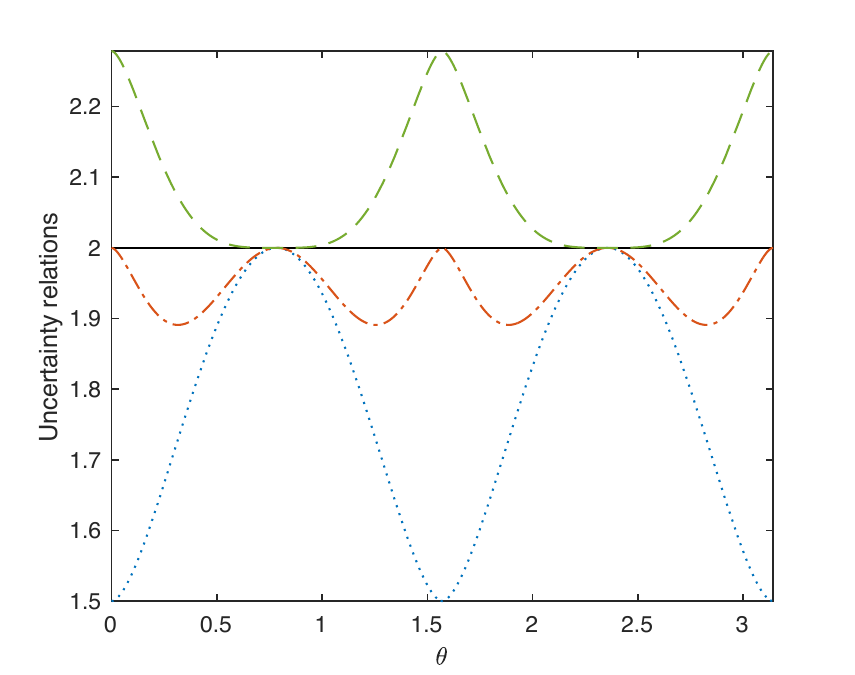}
		 \caption{Comparisons between our uncertainty relations and (\ref{zhang}) in the scenario of three memories. The black (solid) curve represents the uncertainty. The blue (dotted) and red (dot-dashed) denote the lower bounds in (\ref{zhang}) and Theorem~\ref{qmaeurthm1}, respectively. The green (dashed) curve denotes the upper bound of Theorem~\ref{qmaeurthm2}.}
 \label{figex4}
 \end{figure}

\begin{figure}[t]
		 \centering
 \subfigure[]
 {
\label{figex5a}
 \includegraphics[width=4.1cm]{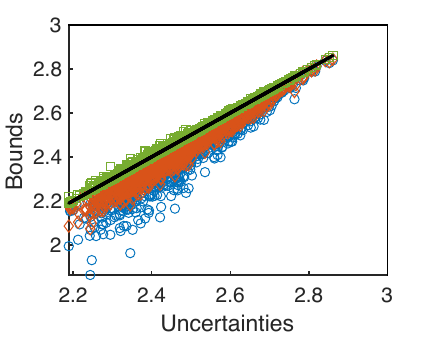}}
 \subfigure[]
 {
  \label{figex5b}
 \includegraphics[width=4.1cm]{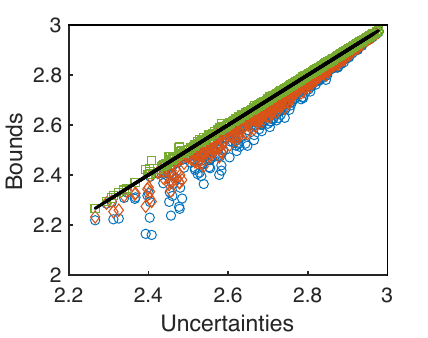}}
 \caption{Comparisons between our uncertainty relations and (\ref{zhang}) in the scenario of three memories with respect to $10^3$ randomly pure states (a) and  $10^3$ randomly mixed states (b). The $y$ axis denotes bounds of the uncertainty. The $x$ axis denotes uncertainty. The blue circles and red diamonds denote the lower bounds of (\ref{zhang}) and Theorem~\ref{qmaeurthm1}, respectively. The green squares denote the upper bound of Theorem~\ref{qmaeurthm2}. }
  \label{figex5}
 \end{figure}

\section{Conclusion}
We have presented two QMA-EURs for CMUBs in the presence of quantum memories and illustrated them in various scenarios involving one, two, and three memories. Although our results outperform existing ones, they are only applicable to dimensions of prime power. Our results given by mutually unbiased bases can be naturally generalized to mutually unbiased measurements (MUMs) for any dimension 
$d$ in the presence of quantum memories. Theoretically, our method may be extended to QMA-EURs with respect to observables or general measurements, provided that a tighter lower bound for the corresponding EURs without quantum memories can be established.

While QMA-EURs for two measurements can be directly applied to the security analysis in quantum key distribution (QKD) protocols employing two conjugate observables~\cite{berta2010uncertainty}, our QMA-EURs have the potential to quantify the secure key rate in multipartite QKD scenarios. Experimentally, uncertainty relations without quantum memory have been demonstrated by using conventional quantum platforms such as trapped ions~\cite{zhang2013state,duan2010colloquium} and photonic qudits~\cite{lapkiewicz2011experimental}. In the presence of quantum memory, experimental demonstrations have also been achieved with more than two measurement settings~\cite{li2011experimental,prevedel2011experimental,PhysRevA.110.062220,PhysRevA.101.032101}. Hence, the uncertainty relations we derived are also experimentally feasible.

\bigskip
\noindent{\bf Acknowledgments}\, \,
We are grateful to X. Ma for her valuable comments on the paper. This work is supported by the National Natural Science Foundation of China (NSFC) (Grant Nos.~12526648, 12171044), the specific research fund of the Innovation Platform for Academicians
of Hainan Province, Natural Science Foundation of Hunan Province (Grant No.~2025JJ60025), Scientific Research Project of the Education Department of Hunan Province (Grant No.~24B0298) and Changsha University of Science and Technology (Grant No.~097000303923).

\bibliography{qmaeur}

\end{document}